\documentclass{new_tlp}
\usepackage{amsmath}
\usepackage{hyperref}
\let\oldproof\proof
\let\oldendproof\endproof
\let\proof\relax
\let\endproof\relax
\usepackage{amsthm}
\usepackage{amsfonts}
\let\proof\oldproof
\let\endproof\oldendproof
\usepackage{comment}
\usepackage{xspace}
\usepackage{tikz}
\usepackage{subcaption}
\usepackage{enumitem}

\tikzset{place/.style={circle, draw}}
\tikzset{>=stealth, auto, node distance=2.5cm, every loop/.style={->, min distance=10mm, in=0, out=60, looseness=10}}
\usepackage{listings}
\usepackage{IEEEtrantools}
\definecolor{darkred}{rgb}{0.5,0.0,0.1}

\title[Modular Answer Set Programming as a Formal Specification Language]{Modular Answer Set Programming as a \\ Formal Specification Language}

\author[Pedro Cabalar, Jorge Fandinno and Yuliya Lierler]
{%
PEDRO CABALAR\\
University of Corunna, Spain\\
\email{cabalar@udc.es}
\and
JORGE FANDINNO
\\
University of Potsdam, Germany\\
\email{fandinno@uni-potsdam.de}
\and
YULIYA LIERLER
\\
University of Nebraska Omaha, USA\\
\email{ylierler@unomaha.edu}%
}

\newtheorem{theorem}{Theorem}
\newtheorem{proposition}{Proposition}

\def\sqht{\hbox{\bf SQHT$^=$}\xspace}

\def\SM{\hbox{\rm SM}}

\def\body{\mathit{Body}}

\def\mp{{\Pi}}
\newcommand{\MP}[1]{\ensuremath{\mp_{#1}}}
\newcommand{\M}[1]{\ensuremath{M_{#1}}}
\newcommand{\stat}[1]{S$_{#1}$}
\def\mh{\ensuremath{\mp_{hc}}}
\def\fh{\ensuremath{\mathbb{H}}}
\def\fh{P_1}

\def\rar{\rightarrow}
\def\lrar{\leftrightarrow}
\def\beq{\begin{equation}}
\def\eeq#1{\label{#1}\end{equation}}
\def\ba{\begin{array}}
\def\ea{\end{array}}

\def\hc{{Hamiltonian cycle}\xspace}
\def\hcs{{Hamiltonian cycles}\xspace}

\def\definition{\mbox{def-module}\xspace}
\def\definitions{{\definition}s\xspace}
\def\modules{\mathit{defmods}}
\def\flattable{flattable\xspace}

\newcommand{\flatt}[1]{\mathit{flat}(#1)}

\def\ANF{$\alpha$-NF}
\def\intens{\mathit{int}}

\def\cP{\mathcal{P}}

\def\cI{\mathcal{I}}
\def\cJ{\mathcal{J}}
\def\cS{\mathcal{S}}
\def\cM{\mathcal{M}}

\newcommand\defm[2]{(#1:#2)}

\newcommand{\eqdef}{%
  \mathrel{\vbox{\offinterlineskip\ialign{%
    \hfil##\hfil\cr%
    $\scriptscriptstyle\mathrm{def}$\cr%
    \noalign{\kern1pt}%
    $=$\cr%
    \noalign{\kern-0.1pt}%
}}}}

\def\pred{\mathit{pred}}
\def\formula{\mathcal{F}}

\newtheorem{prop}{Proposition}

\newtheorem{lemma}{Lemma}

\newcommand\restr[2]{#1_{|#2}}
\renewcommand\vec[1]{{\bf #1}}

\newcommand{\set}[1]{\ensuremath{\{#1\}}}

\newcommand{\setm}[2]{\ensuremath{\{\ #1\ |\ #2\ \}}}

\newcommand{\tuple}[1]{\ensuremath{\langle #1 \rangle}}

\newcommand{\pedro}[1]{\colorlet{saved}{.}\color{blue}{\it\small(P: #1)}\color{saved}\xspace}

\newcommand{\citep}{\cite}
\newcommand{\citet}[1]{\citeauthor{#1}~[\citeyear{#1}]}

\makeatletter
\let\proof\relax
\let\endproof\relax
\newenvironment{proof}[1][\emph{Proof}]{\par
  \pushQED{\qed}%
  \normalfont \topsep6\p@\@plus6\p@\relax
  \trivlist
  \item[\hskip\labelsep
    #1\@addpunct{.}]\ignorespaces
}{%
  \popQED\endtrivlist\@endpefalse
}

\newenvironment{proofs}[1][\emph{Proof sketch}]{\par
  \pushQED{\qed}%
  \normalfont \topsep6\p@\@plus6\p@\relax
  \trivlist
  \item[\hskip\labelsep
    #1\@addpunct{.}]\ignorespaces
}{%
  \popQED\endtrivlist\@endpefalse
}
\makeatother

\def\vertex{{\it vertex}}
\def\edge{{\it edge}}
\def\inp{{\it in}}
\def\rp{{\it r}}
\def\rap{{\it ra}}

\def\COMP{\hbox{\rm COMP}}
\def\CIRC{\hbox{\rm CIRC}}
\begin{document}

\clearpage

\maketitle

\begin{abstract}
In this paper, we study the problem of formal verification for Answer Set Programming (ASP), namely, obtaining a \emph{formal proof} showing that the answer sets of a given (non-ground) logic program $P$ correctly correspond to the solutions to the problem encoded by $P$, regardless of the problem instance. To this aim, we use a formal specification language based on ASP modules, so that each module can be proved to capture some informal aspect of the problem in an isolated way. This specification language relies on a novel definition of (possibly nested, first order) \emph{program modules} that may incorporate local hidden atoms at different levels. Then, \emph{verifying} the logic program~$P$ amounts to prove some kind of equivalence between~$P$ and its modular specification.
\end{abstract}

\section{Introduction}

Achieving trustworthy AI systems requires, among other qualities,
the assessment that those systems produce correct judgments%
\footnote{\emph{Ethics Guidelines For Trustworthy AI}, High-level Expert Group on Artificial Intelligence set up by the European Commission.
\\
\url{https://ec.europa.eu/digital-single-market/en/news/ethics-guidelines-trustworthy-ai}.
}
or, in other words, the ability to  verify that produced results adhere to specifications on expected solutions. 
These specifications may have the form of expressions in some formal language or may  amount to statements in natural language (consider English used in mathematical texts).
Under this trust-oriented perspective, AI systems built upon some Knowledge Representation (KR) paradigm start from an advantageous position, since their behavior is  captured by some declarative machine-interpretable formal language.
Moreover, depending on its degree of declarativity, a KR formalism can also be seen as a specification language by itself.
%

Answer Set Programming~(ASP;~\citeNP{niemela99a,martru99a})
is a well-established KR paradigm for solving knowledge-intensive search/optimization problems.
Based on logic programming under the \emph{answer set semantics}~\cite{gellif88b},
the ASP methodology relies on devising a logic program so that its answer sets are in one-to-one correspondence to the solutions of the target problem.
This approach is fully declarative, since the logic program only describes a problem and conditions on its solutions, but abstracts out the way to obtain them, delegated to systems called answer set solvers.
%
%
Thus, it would seem natural to consider
an ASP program to serve a role of a formal specification on its expected solutions.
%
%
However, 
the non-monotonicity of the ASP semantics makes it difficult to directly associate an independent meaning to an arbitrary program fragment (as we customary do, for instance, with theories in classical logic).
This complicates assessing that a given logic program reflects, in fact, its intended \emph{informal description}.
And yet, ASP practitioners do build logic programs in groups of rules and identify each group with some part of the informal specification~\cite{erdo04,lifschitz17a}.
Moreover, modifications on a program frequently take place within a group of rules rather than in the program as a whole.
The safety of these local modifications normally relies on such properties in ASP  as \emph{splitting}~\cite{liftur94a,feleli11a}.
%

With these observations at hand, 
we propose a \emph{verification methodology for logic programs}, where the argument of correctness of a program is decomposed into respective statements of its parts, relying to this aim on a modular view of ASP in spirit of~\cite{OikarinenJ09,har16}.
In this methodology, a \emph{formal specification}~$\Pi$ is formed by a set of modules called \emph{modular program}, so that each module (a set of program rules) is {\em ideally small enough to be naturally related to its informal description} in isolation.
This relation can be established using quasi-formal English statements~\cite{den12,lifschitz17a} or relying on classical logic~\cite{lifschitz18a}.
The same specification $\Pi$ may serve to describe different (non-modular) ASP programs encoding the same problem. Each such program $P$ respects given priorities involving efficiency, readability, flexibility, etc.
As usual in \emph{Formal Methods}~\cite{mon03}, \emph{verification} then consists in obtaining a \emph{formal proof} of the correspondence between the verified object (in our case, the non-modular encoding $P$) and its specification (the set $\Pi$ of modules matching the informal aspects of that problem).
%
%
It is important to note that the formal specification language used is subject to two important requirements: (i) dealing with  {\em non-ground} programs; and (ii) capturing stable models of these programs using an expression that can be formally manipulated.
For (i), we could  use \emph{Quantified Equilibrium Logic}~\cite{PearceV08} but for (ii) the equivalent formulation in~\cite{ferraris11a} is more suitable, as it captures stable models of a program as a second-order logic formula,%
\!\footnote{The need for second-order logic is not surprising:  the stable models of a logic program allow us to capture a transitive closure relation.} 
the SM operator.

%
Once a modular specification is guaranteed (related to its informal description/proved to be correct with respect to its informal description), we expect that 
 arguing the correctness of the replacement of some of its module by a different encoding is reduced to arguing some kind of equivalence between modules without affecting the rest of the correctness proof.
%
%
%
The difficulty here appears when \emph{auxiliary predicates} are involved.
These predicates are quite common to improve the solver performance in a given encoding $P$ but, more importantly, they are sometimes indispensable in the specification $\Pi$ to express some property that the ASP language cannot directly capture otherwise~\cite{goncalves16a}. 
In both cases, their presence must be taken into account when proving correctness which, in its turn, normally depends on their \emph{local use in some part of the program}.

In this paper, we extend the modular language from~\cite{har16} to allow for  a \emph{hierarchical tree} of modules and submodules, each of them possibly declaring its own set of {\em public} and {\em hidden} predicates at different levels.
%
%
We use this extension as a language for formal specifications. For illustration, we consider a logic program encoding the well known Hamiltonian Cycle (HC) problem.
We start by providing a formal specification of  the HC problem using the hierarchy of modules and relate it to the informal description of the problem. We then formally prove the correspondence between the HC logic program  and its hierarchical specification.
This constitutes the argument of correctness for the considered logic program.
We also provide an example of module replacement that is verified through an equivalence proof that disregards the existing auxiliary predicates.

\noindent
{\em Paper outline:} Section~\ref{sec:motivation} provides a running example of the HC problem encoding and presents our methodology. 
 In Section~\ref{sec:sm}, we revisit the SM operator and extend it with hidden predicates.
 Section~\ref{sec:modlog} presents our modular logic programs, while 
 Sections~\ref{sec:specification} and~\ref{sec:verif} explain their use for formal verification, illustrating the proposed methodology on the running example.
\begin{figure}[t]
\begin{subfigure}{.3\textwidth}
\centering
\begin{tikzpicture}[scale=1, node distance=2cm and 0.1cm, every node/.style={scale=1}]
\node[place] (a) {{\it a}};
\node[place] (b)  [right of=a] {{\it b}};
\node[place] (c) [below  of=b] {{\it c}};
\node[place] (d) [below  of=a] {{\it d}};
\draw [->,double distance=2pt] (a) to  (b);
\draw [->,double distance=2pt] (b) to  (c);
\draw [->,double distance=2pt] (d) to  (a);
\draw [->,double distance=2pt] (c) to [out=155,in=25] (d);
\draw [->] (d) to [out=-25,in=205] (c);
\end{tikzpicture}
\caption{Graph $G_1$}
\label{fig:example_g1}
\end{subfigure}
\begin{subfigure}{.5\textwidth}
\centering
\begin{tikzpicture}[shorten >=2pt, node distance=1.6cm and 0.1cm, auto,scale=1, every node/.style={scale=1}]
\node[place] (e) {{\it edge}};
\node[place] (v)  [right of=e] {{\it vertex}};
\node[place] (i) [below  of=v] {$\;\;${\it in}$\;\;$ };
\node[place] (r) [right  of=i] {$\;\;\;${\it r}$\;\;$ };
\path[->]
 (r) edge [loop above] node {} ()
 (v) edge[->] node[above] {} (e)
 (i) edge[->] node[below] {} (e)
 (r) edge[->] node[above] {} (i);
\end{tikzpicture}
\caption{Dependency graph for $\mh^m$}
\label{fig:example_dep}
\end{subfigure}
\caption{A pair of graphs used in the examples.}
\label{fig:example}
\end{figure}
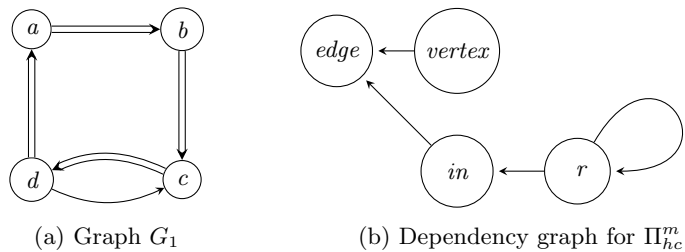
\section{Motivating example and methodology}\label{sec:motivation}

We consider 
a well-known domain in the ASP literature: the search of Hamiltonian cycles in a graph.
A \hc is a cyclic path from a directed graph that visits each of the graph's vertex exactly once.
For instance, Figure~\ref{fig:example_g1} depicts graph~$G_1$, whose unique Hamiltonian cycle is marked in double lines, whereas Listing~\ref{list:hamiltonian} presents a possible encoding of the problem in the language of ASP solver {\sc clingo}~\citep{geb07}.
The {\tt \#show} directive is used to tell {\sc clingo} which predicates (we call these \emph{public}) should appear in the obtained answer sets: in this case, only predicate {\tt in/2} that captures the edges in the solution. 
Now, if we add the facts in Listing~\ref{list:graph1} corresponding to graph $G_1$ and instruct {\sc clingo} to find all the answer sets we obtain the output in Listing~\ref{list:answer}, which is the only Hamiltonian cycle in the graph.
ASP practitioners usually explain Listing~\ref{list:hamiltonian} in groups of rules.
Rule~3 is used to generate all possible subsets of edges, while Rules~\mbox{7-8} guarantee that connections among them are linear.
Rules~\mbox{4-5} are meant to define the auxiliary predicate~$\mathit{r}$ (for ``reachable'') as the transitive closure of predicate~$\mathit{in}$.
To assign this meaning to~$r$, we \emph{implicitly assume} that this predicate does not occur  in other rule heads in the rest of the program.
Predicate $r$ is then used in  Rule~6 to enforce that any pair of vertices are connected.
Rules~\mbox{4-6} together guarantee that facts for $in$ form a strongly connected graph covering all vertices in the given graph.
\begin{figure}[t]
\begin{lstlisting}[
  caption = {Encoding of a Hamiltonian cycle problem using {\sc clingo}. 
%   The program is annotated with \emph{achievements}.
},
  label={list:hamiltonian},
  basicstyle=\ttfamily\small,
  numbers=left,
  stepnumber=2,
]
vertex(X):- edge(X,Y). 
vertex(X):- edge(Y,X).
{ in(X,Y) }:- edge(X,Y).
r(X,Y):- in(X,Y).  
r(X,Y):- r(X,Z), r(Z,Y).
:- not r(X,Y), vertex(X), vertex(Y).
:- in(X,Y), in(X,Z), Y != Z.
:- in(X,Y), in(Z,Y), X != Y.
#show in/2.
\end{lstlisting}
%
\begin{lstlisting}[
  caption = {Facts describing graph $G_1$.
%   The program is annotated with \emph{achievements}.
  },
  label={list:graph1},
  basicstyle=\ttfamily\small
]
edge(a,b).  edge(b,c).  edge(c,d).  edge(d,a).  edge(d,c).
\end{lstlisting}
%
\begin{lstlisting}[
caption = {Output by {\sc clingo} for program composed of lines in Listing~\ref{list:hamiltonian} and~\ref{list:graph1}.},
label={list:answer},
  basicstyle=\ttfamily\small
]
Answer: 1
in(a,b) in(b,c) in(c,d) in(d,a).
\end{lstlisting}
%
\begin{lstlisting}[
  caption = {Alternative code to lines~\mbox{4-6} in Listing~\ref{list:hamiltonian}~\protect\cite{martru99a}.
%   The program is annotated with \emph{achievements}.
},
  label={list:hamiltonian2},
  basicstyle=\ttfamily\small,
  stepnumber=2,
]
ra(Y) :- in(a,Y).
ra(Y) :- in(X,Y), ra(X).
:- not ra(X), vertex(X).
\end{lstlisting}
\end{figure}
\paragraph{Methodology.}
The methodology we propose for verifying the correctness of some logic program under  answer set semantics consists of the following steps:
\begin{enumerate}[label=Step~\Roman*., ref=Step~\Roman*, labelwidth=38pt, align=left, leftmargin=45pt]
\item Decompose the informal description of the problem into independent (natural language) statements $S_i$, identifying their possible hierarchical organization.
    \label{item:methodology.1}

\item Fix the public predicates used to represent the problem and its solutions.
\label{item:methodology.2a}

\item Formalize the specification of the statements as a non-ground modular program~$\Pi$, possibly introducing (modularly local) auxiliary predicates.
    \label{item:methodology.2b}

\item Construct an \emph{argument} (a ``metaproof'' in natural language) for the correspondence between $\Pi$ and the informal description of the problem.
    \label{item:methodology.3}
\item Verify that the given logic $P$ program adheres to the formal specification $\Pi$. The result of this step is a set of formal proofs.
    \label{item:methodology.4}
\end{enumerate}
Note that the first four steps are exclusively related to the formal specification $\Pi$ of the problem, while the particular program $P$ to be verified is only considered in~\ref{item:methodology.4}, where formal verification proofs are produced.

Now, back to Hamiltonian cycle problem, a possible and reasonable result of~\ref{item:methodology.1} is the hierarchy of statements:
\begin{enumerate}
\item A Hamiltonian cycle~$G'$ of graph~$G$ must be a subgraph of~$G$ that contains all vertices of~$G$, that is:
\label{herarchy:1}
    \begin{enumerate}[topsep=1pt]
    \item $G'$ has the same vertices as~$G$, and
    \label{herarchy:1a}
    \item all edges of~$G'$ also belong to~$G$.
    \label{herarchy:1b}
    \end{enumerate}
\item A Hamiltonian cycle~$G'$ of graph~$G$ is a cycle that visits all vertices of~$G$ exactly once, that is:
\label{herarchy:2}
    \begin{enumerate}[topsep=1pt]
    \item no vertex has more than one outgoing/incoming edge on~$G'$, and
    \label{herarchy:2a}
    \item $G'$ is strongly connected.
    \label{herarchy:2b}
    \end{enumerate}
\end{enumerate}
%
The choice for public predicates (\ref{item:methodology.2a}) is, of course, arbitrary, but must be decided to compare different encodings (as also happens, for instance, when we fix a benchmark).
Here, we choose predicates $edge/2$ and $in/2$ to encode the edges of the input graph $G$ and the Hamiltonian cycle $G'$, respectively. 
To prove the correctness of the encoding in Listing~\ref{list:hamiltonian}, we resume the rest of our methodological steps later on, in Sections~\ref{sec:specification} and~\ref{sec:verif}.
For instance, \ref{item:methodology.2b} is shown in Section~\ref{sec:specification}, where we define a formal specification~$\MP{1}$ that happens to comprise the same rules as Listing~\ref{list:hamiltonian} but for the one in line~8.
The main difference is that rules in $\MP{1}$ are grouped in modules corresponding to the above hierarchy.
A set of propositions in Section~\ref{sec:specification} are used to establish the correspondence between $\MP{1}$ (\ref{item:methodology.3}) and the informal statements. 
The already mentioned strong relation between Listing~\ref{list:hamiltonian} and $\MP{1}$ is not something we can always expect.
As happens with \emph{refactoring} in software engineering, encodings usually suffer a sequence of modifications to improve some of their attributes (normally, a better efficiency) without changing their functionality.
Each new version implies a better performance of the answer set solver, but its correspondence with the original problem description becomes more and more obscure~\cite{bud15}.
For instance, it might be noted that program in Listing~\ref{list:hamiltonian} produces an excessively large ground program when utilized on graphs of non trivial size.
It turns out that it is enough to require that all vertices are reachable from some fixed node in the graph (for instance~$a$).
%
Thus, rules~\mbox{4-6} of Listing~\ref{list:hamiltonian} are usually replaced by rules in Listing~\ref{list:hamiltonian2}. The answer sets with respect to predicate $in$ are identical, if the graph contains a vertex named~$a$. 
In this sense, \emph{verifying} an ASP program can mean establishing some kind of \emph{equivalence} result between its formal specification in the form of a modular program and the final program obtained from the refactoring process~\cite{lier19}.
This is tackled in Section~\ref{sec:verif} and constitutes~\ref{item:methodology.4}.
\section{Operator SM with hidden predicates}\label{sec:sm}

Answer set semantics has been extended to arbitrary first-order~(FO) theories with the introduction of \emph{Quantified Equilibrium Logic}~\cite{PearceV08} and its equivalent formulation using the second-order (SO) operator {SM}~\cite{feleli11a}.
These approaches allow us to treat program rules as logical sentences with a meaning that bypasses grounding.
For instance, rules in Listing~\ref{list:hamiltonian} respectively correspond to:
\begingroup
\allowdisplaybreaks
\begin{align}
&\forall x y({\it edge}(x,y) \rar {\it vertex}(x)) \; \label{eq:hce1}\\
&\forall x y({\it edge}(y,x) \rar {\it vertex}(x)) \; \label{eq:hce2}\\
&\forall x y( \neg\neg  {\it in}(x,y) \wedge {\it edge}(x,y) \rar {\it 
in}(x,y)) \; \label{eq:choicerulefo}\\
&\forall x y({\it in}(x,y) \rar {\it r}(x, y)) \; \label{eq:hcr1}\\
&\forall x y z({\it r}(x,z) \wedge {\it r}(z,y)\rar {\it r}(x,y)) \; \label{eq:hcr2}
\\
&\forall x y (\neg {\it r}(x, y)\wedge {\it vertex}(x)\wedge {\it 
vertex}(y)\rar \bot) \label{eq:hcmod4}\\
&\forall x y z({\it in}(x,y)\wedge {\it in}(x,z)\wedge \neg(y=z) \rar\bot) 
\; \label{eq:hcmod2.1}\\
&\forall x y z({\it in}(x,z)\wedge {\it in}(y,z)\wedge \neg(x=y) \rar\bot) 
\; \label{eq:hcmod2.2} 
\end{align}
\endgroup
As we can see, the correspondence is straightforward except, perhaps, for the \emph{choice} rule in line~3
of Listing~\ref{list:hamiltonian} that is represented as formula~\eqref{eq:choicerulefo} (see~\citeNP{fer05} for more details).
We name the conjunction of sentences~\eqref{eq:hce1}-\eqref{eq:hcmod2.2} as our encoding $\fh$.

We now recall the SM operator from~\cite{feleli11a}, assuming some familiarity with second order~(SO) logic.
%
We adopt some notation: a letter in boldface inside a formula denotes a tuple of elements also treated as a set, if all elements are different.
Quantifiers with empty tuples (or empty sets of variables) can be removed: \mbox{$\exists \emptyset F = \forall \emptyset F := F$}.
Expression $pred(F)$ stands for the set of free predicate names (different from equality) in a SO formula $F$.
If $p$ and $q$ are predicate names of the same arity~$m$ then $p \leq q$ is an 
abbreviation of  
$
\forall \vec{x}(p(\vec{x}) \rar q(\vec{x})), 
$
where $\vec{x}$ is an $m$-tuple of \mbox{FO variables}.
Let $\vec{p}$ and $\vec{q}$ be tuples $p_1, \dots, p_n$ and $q_1, \dots, q_n$ of predicate symbols or variables.
Then 
$\vec{p} \leq \vec{q} :=
(p_1 \leq q_1) \land \dots \land (p_n \leq q_n),
$ 
and 
$\vec{p} < \vec{q}$ is an abbreviation of 
$(\vec{p} \leq \vec{q}) \land \neg (\vec{q} \leq \vec{p})$.
Given a FO formula $F$, its {\em stable model operator with intensional predicates $\vec{p}$} (not including equality) is the SO formula 
\begin{equation}
\SM{\vec{p}}[F] \quad := \quad F \land \neg \exists \vec{U}\big( \ (\vec{U} < \vec{p}) \land F^*\ \big), \label{f:sm}
\end{equation}
with $\vec{U} = U_1, \dots, U_n$ distinct predicate variables not occurring in $\vec{p}$ and:
\[
F^* := \left\{
\begin{array}{cl}
F & \text{if $F$ is an atomic formula without members of } \vec{p}\\
U_i(\vec{t}) & \text{if $F$ is $p_i(\vec{t})$ for $p_i \in \vec{p}$ and $\vec{t}$ a tuple of terms}\\
G^* \otimes H^* & \text{if } F = (G \otimes H) \ \text{and } \otimes \in \{\wedge, \vee\}\\
(G^* \to H^*) \wedge (G \to H) & \text{if } F = (G \to H) \\
Qx \; (G^*) & \text{if } F = Qx\; G \ \text{ and } Q \in \{\forall, \exists\}\\
\end{array}
\right.
\]
Predicate symbols occurring in $F$ but not in $\vec{p}$ are called \emph{extensional} and are interpreted classically.
In fact, if $\vec{p}$ is empty, it is easy to see that SM$_\vec{p}[F]$ coincides with $F$.
%
%
We say that an interpretation~$\cI$ over a signature~$\cP$ is an \emph{answer set} of a FO formula~$F$ (representing a logic program) when it is a Herbrand model of
SM$_\vec{p}[F]$ and $\vec{p}=pred(F)$.
When $F$ is a logic program, answer sets defined in this way correspond to the traditional definition based on grounding~\cite{feleli11a}.
It is common to identify Herbrand interpretations with sets of atoms
corresponding to its predicates and their extensions.
For a Herbrand interpretation $\cI$ over set~$\cP$ of predicates  and set $\cS\subseteq\cP$, we write  $\restr{\cI}{\cS}$ to denote the restriction of $\cI$ to $\cS$.
As usual, the \emph{extent} of a predicate $p$ in interpretation~$\cI$, written $p^\cI$, collects every tuple of Herbrand terms $\vec{t}$ for which $p(\vec{t})$ holds in $\cI$.

As a small example, let $F_g$ be the conjunction of facts in Listing~\ref{list:graph1}, $F_1$ denote a conjunction of  $F_g$ and $\eqref{eq:hce1}$ and let $\vec{U}$ be $\tuple{E,V}$. Then, $F_1^*$ is formed by the conjunction of
$
\forall x y ( (E(x,y) \to V(x)) \wedge (edge(x,y) \to vertex(x)))
$
and all atoms $E(x,y)$, one per each fact $edge(x,y)$ in $F_g$.
The answer sets of $F_1$ are captured by the Herbrand models of $\text{SM}_{\tuple{edge,vertex}}[F_1]$.
This formula has a unique model that minimizes the extension of $edge$ to the exact set of facts in $F_g$ (and not more) and the extension of $vertex$ to be precisely all nodes used as left arguments in those edges.
If we take instead the formula $F_2:=\fh \wedge F_g$
and $\vec{q}=\mathit{pred}(F_2)$ then 
\mbox{SM$_{\vec{q}}[F_2]$}
has a unique Herbrand model that has the same atoms for predicate ${\it in}$ as those in Listing~\ref{list:answer} but, obviously, has also more atoms for the remaining predicates in $\mathit{pred}(F_2)$. 
A simple way of removing (or forgetting) those extra predicates in SO is adding their existential quantification.
Given formula $F$ we define the \emph{answer sets of~$F$ for $\vec{p}$ hiding predicates $\vec{h}$} as the Herbrand models of:
\begin{equation}
\exists \vec{H} \ \text{SM}_\vec{p}[F^\vec{h}_\vec{H}] \label{f:hSM}
\end{equation}
where $\vec{H}$ is a tuple of predicate variables of the same length as~$\vec{h}$
and $F^\vec{h}_\vec{H}$ is the result of replacing all occurrences of predicate symbols from~$\vec{h}$ by the corresponding predicate variables from~$\vec{H}$.
We abuse the notation and write $\exists \vec{h} \ \text{SM}_\vec{p}[F]$ instead of~\eqref{f:hSM} when it does not lead to confusion.
We call predicates in tuples $\vec{p}$ and  $\vec{h}$ {\em intensional} and {\em hidden}, respectively.
For instance, 
$\exists \mathit{edge} \, \mathit{vertex} \, \mathit{r} \  \text{SM}_{\vec{q}}[F_2]$
stands for the formula
$\exists E\,V\,R \  \text{SM}_{\vec{q}}[F_2']$
where~$F_2'$ is obtained from~$F_2$ by replacing predicate symbols $\mathit{edge},\, \mathit{vertex},\, \mathit{r}$ by variables~$E,\,V,\,R$; and it has a unique Herbrand model that coincides with the one in Listing~\ref{list:answer}.
Thus, this model forms the unique answer set of $F_2$ for $\vec{q}$ hiding predicates
$\mathit{edge}, \mathit{vertex}$ and $\mathit{r}$.
This corresponds to the behavior produced by the~$\tt \#show$ directive in Line~9 of the \hc encoding.
%
The following proposition states that the existential quantifiers for hidden predicates filter out their information in the models.

\begin{proposition}\label{thm:sm-traditional}
Let $F$ be a formula, $\vec{h}$ a tuple of predicate symbols from $F$ and let
\mbox{$\cS=pred(F) \setminus \vec{h}$}.
Then, $\cI$ is a Herbrand model of \eqref{f:hSM} iff there is some answer set $\cJ$ of~$F$ such that $\cI=\restr{\cJ}{\cS}$.
\end{proposition}
%

\section{Nested modular programs}\label{sec:modlog}

As explained in the introduction, our specification language departs from the work on modular logic programs by~\citeN{har16}, where each module is a FO formula whose semantics is captured by the SM operator.
This choice is motivated by two factors.
First, it allows treating programs or modules with variables as FO formulas that can be manipulated without resorting to the program grounding at all. 
This is crucial is we wish to argue about the meaning of each module regardless of the particular instance of the problem to be solved.
Second, the SM operator captures the semantics of the program also as a formula that, although including SO quantifiers, can be formally treated using a manifold of technical results from the literature (including  theorems on splitting or constraints, for instance).
To adapt this modular approach to our purposes in this paper,
we extend it in two ways:
(i)~handling auxiliary predicates inside modules
and (ii) introducing nested modules.

%
A \emph{\definition} $M$ is a pair \mbox{$(\vec{p}:F)$}, where $\vec{p}$ is a tuple of intensional predicate symbols and $F$ is a \mbox{FO formula}.
Its  semantics  is captured by the formula $\Phi(M) \eqdef \SM_\vec{p}[F]$.
When $\vec{p}$ is empty, we write~$F$ in place of \definition \mbox{$(\vec{p}:F)$} --- indeed,  $\Phi(M)$ amounts to $F$.
On the contrary, when $\vec{p}=pred(F)$ all predicates in $F$ are intensional, so \definition \mbox{$\defm{pred(F)}{F}$} represents the usual situation in a logic program $F$.
%
If we leave some predicates as extensional, then their extent is not ``minimized.''
For instance, rule $\eqref{eq:hce1}$ alone, i.e., the \definition $\defm{edge,vertex}{\eqref{eq:hce1}}$, has a unique answer set $\emptyset$ whereas in a \definition~$\defm{vertex}{\eqref{eq:hce1}}$
no assumption is made on extensional predicate $edge$, while $vertex$ collects precisely  all left arguments of $edge$.
%

A \emph{modular program} is a pair ${\mp = \tuple{\cS,\cM}}$, where $\cS \subseteq \cP$ is the set of \emph{public} predicate symbols and $\cM$ is the set of \emph{modules} $\{M_1,\dots,M_n\}$, where $M_i$ ($1\leq i\leq n$) is either a \emph{\definition} or another modular program (called  \emph{subprogram} of $\mp$), so that all predicate symbols occurring in $\mp$ are in $\cP$.
A modular program can be depicted as a hierarchical tree, whose leafs are \definitions.
The interpretation of a modular program $\mp=\tuple{\cS, \cM}$ is captured by the (recursively defined) formula:
\[
\Phi(\mp)
	\ \ \eqdef \ \ \exists \vec{h} \
\bigwedge \left\{\Phi(M_i) \mid M_i \in \cM\right\}
\]
where $\vec{h}$ contains all free predicate symbols in $\Phi(M_i)$ that are not in $\cS$.
Notice that, if~$M_i$ is a \definition $\defm{\vec{p}}{F}$, then $\Phi(M_i)=\SM_\vec{p}[F]$ as we defined before, while if $M_i$ is a subprogram, we apply the formula above recursively.
%
We say that an interpretation~$\cI$ is a \emph{model} of a modular program~$\mp$,
when it satisfies the formula~$\Phi(\mp)$.
As in the previous section,
we say that interpretation~$\cI$ over public predicate symbols $\cS$  is an \emph{answer set} of a modular program~$\mp$ when it is a Herbrand model of~$\mp$.
Programs $\mp$ and $\mp'$ are said to be \emph{equivalent} if they have the same answer sets.
Under logic programming syntax (like Listing~\ref{list:hamiltonian}), public predicates are declared via {\tt \#show} clauses. 
We sometimes allow the \definition $M=(\vec{p}:F)$ as an abbreviation of the modular program $\tuple{pred(F), \set{M}}$.

We define $\modules(P)$ as the set of all \definitions in $\mp$ at any level in the hierarchy.
\mbox{Program~$\mp=\tuple{\cS,\cM}$} is said to be \emph{flat} when it does not contain subprograms, i.e., $\cM$ coincides with $\modules(P)$.
We call flat modular program $\tuple{\cS,\cM}$ \emph{\mbox{HL-modular}}, when 
its set~$\cS$ of public predicates contains all predicate symbols occurring in $\cM$. HL-modular programs capture the definition of modular programs from~\citep{har16}.

To illustrate these definitions, we provide a modular program $\MP{1}$ that will act later on as a specification for encoding~$\fh$=~\eqref{eq:hce1}-\eqref{eq:hcmod2.2} from Section~\ref{sec:sm}. 
Program $\MP{1}$ comprises the same rules~\eqref{eq:hce1}-\eqref{eq:hcmod2.1} (as we will see, \eqref{eq:hcmod2.2} is actually redundant) but organizes them in the  tree of Figure~\ref{fig:tree}.
The tree contains 5 \definitions (the leaf nodes) and has three subprograms, $\MP{sb}$, $\MP{hc}$ and $\MP{cn}$ at different levels.
Each modular program node (drawn as a thick line box) also shows inside its set of public predicates.
Note, for instance, how predicate $r$ is local to subprogram $\MP{cn}= \tuple{\{vertex,in\}, \{\big({\it r} : \eqref{eq:hcr1} \wedge \eqref{eq:hcr2} \big), \eqref{eq:hcmod4}\}}$ that corresponds to rules~4-6 in Listing~\ref{list:hamiltonian} and intuitively states that relation~$in$ forms a strongly connected graph covering all vertices.
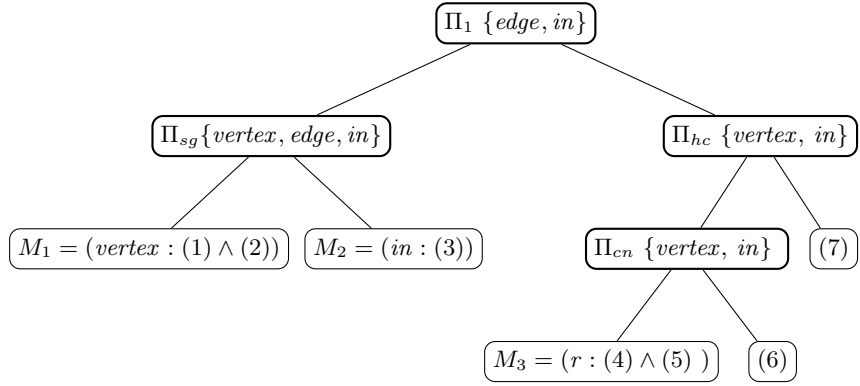
\begin{figure}
\centering
\begin{tikzpicture}[
  level 1/.style={sibling distance=20em},
  level 2/.style={sibling distance=10em},
  level 3/.style={sibling distance=7em},
  level 4/.style={sibling distance=10em},
  dm/.style={thin},
  every node/.style = {thick,shape=rectangle,draw,rounded corners, align=center}]
  \node {$\MP{1}$ $\{\mathit{edge}, \mathit{in}\}$}
    child { node {$\MP{sg} \{\mathit{vertex}, \mathit{edge}, \mathit{in}\}$} 
        child { node[dm] {$\M{1} = (\mathit{vertex}: \eqref{eq:hce1} \wedge \eqref{eq:hce2})$} }
        child { node[dm] {$\M{2} = ({\it in} : \eqref{eq:choicerulefo})$} }
    }
    child { node {$\MP{hc} \ \{{\it vertex,\, in}\}$}
      child[sibling distance=6em] { node {$\MP{cn} \ \{{\it vertex,\, in}\}$ }
        child { node[dm] {$\M{3} = ({\it r} : \eqref{eq:hcr1} \wedge \eqref{eq:hcr2} \ )$} }
        child { node[dm] {$\eqref{eq:hcmod4}$} }
      }
      child[sibling distance=6em] { node[dm] {$\eqref{eq:hcmod2.1}$} } 
    };
\end{tikzpicture}
\caption{Hierarchical structure of modular program $\MP{1}$.}
\label{fig:tree}
\end{figure}

%
%

%

\section{A formal specification language}
\label{sec:specification}

Hamiltonian cycles constitute a good example of a typical use of ASP for solving a search problem.
Following~\cite{bre11}, a {\em search problem}~$X$ can be seen as a set of instances, being each
{\em instance}~$I$ assigned a finite set~$\Theta_{X}(I)$ of solutions. 
Under our verification method, we propose constructing a modular program~$\mp_X$ that adheres to the specifications of $X$ so that when extended with modular program $\mp_I$ representing an instance~$I$ of $X$, the answer sets of this join are in one to one correspondence with members in~$\Theta_X(I)$.
%
%
Then, we can use~$\mp_X$ to argue, module by module, that each of its components actually corresponds to some part of the informal description of the problem.
To illustrate these ideas we prove next that, indeed, program $\MP{1}$ presented in Figure~\ref{fig:tree} and described in the last section is a \emph{formal specification} for the search of Hamiltonian cycles.
%
%
%
%
We start by presenting the informal readings of all \definitions occurring in the program --- i.e., members of  $\modules(\MP{1})$.
The \definitions~\M{1}, \M{2} and \M{3} intuitively formulate the following {\em Statements}:
\begin{enumerate} 
\setlength\itemsep{0em}
\item[\stat{1}:] ``In any model of~\M{1},
the extent of ${\it vertex}$ collects all objects in the extent of ${\it edge}$.''

\item[\stat{2}:] ``In any model of~\M{2},
the extent of ${\it in}$ is a subset of the extent of 
 ${\it edge}$.''
\item[\stat{3}:] ``In any model of~\M{3},
the extent of {\it r} is the transitive closure of the extent of {\it in}.''
\end{enumerate}
These statements can be seen as {\em achievements} of each \definition in the sense of~\cite{lifschitz17a}, 
that are agnostic to the context where \definitions appear.
This closely aligns with good software engineering practices, where the emphasis is made on modularity/independence of code development.
An intuitive meaning of formulas~\eqref{eq:hcmod4} and~\eqref{eq:hcmod2.1} is self explanatory given the underlying conventions of FO logic: 
\begin{enumerate}
\item[\stat{\eqref{eq:hcmod4}}:] ``In any model of~\eqref{eq:hcmod4},
the extent of ${\it r}$ contains each possible pair of vertices.''

\item[\stat{\eqref{eq:hcmod2.1}}:] ``In any model of~\eqref{eq:hcmod2.1},
the extent of ${\it in}$ does not contain two different pairs with the same left component.''
\end{enumerate}
Statements~\stat{1} and~\stat{2} translate into a joint {\em Statement} about module~$\MP{sg}$:
\begin{enumerate}
\item[\stat{sg}:]\label{l:hcmod.sg}
``In any model~$\cI$ of~$\MP{sg}$, $\tuple{\mathit{vertex}^\cI,\mathit{in}^\cI}$ is a subgraph of $\tuple{\mathit{vertex}^\cI,\mathit{edge}^\cI}$.''
\end{enumerate}
Similarly, we identify the following two combined statements:
\begin{enumerate}
\item[\stat{cn}:]\label{l:hcmod7}
``In any model of~$\mp_{cn}$,
the extent of ${\it in}$ forms a strongly connected graph covering all vertices.''

\item[\stat{hc}:]
``In any model of~$\mh$,
the extent of ${\it in}$ is a cycle visiting all vertices exactly once'' or equivalently ``it induces a graph that is a \hc.''
\end{enumerate}
Note that each subcomponent of $\MP{1}$ is small enough so that the verification of its corresponding statement is a manageable and self-isolated task. 
Note also that statements~\stat{sg} and~\stat{hc} are the result of fixing the public vocabulary for the statements \ref{herarchy:1} and~\ref{herarchy:2} identified from the informal description in Section~\ref{sec:motivation}.
Statements~\ref{herarchy:1a} and~\ref{herarchy:1b} in the informal description correspond to statements~\stat{1}
and~\stat{2}, respectively.
Statements~\ref{herarchy:2a} and~\ref{herarchy:2b} correspond to statements~\stat{\eqref{eq:hcmod2.1}}
and~\stat{cn}.

At this point, we have completed~\ref{item:methodology.2b} for our example, with the modular program~$\MP{1}$ and the informal statements to compare with.
Now, \ref{item:methodology.3} consists on building claims about the correctness of the modules versus the statements.
Proofs for the correctness of statements~\stat{1}, \stat{2}, \stat{3}, \stat{\eqref{eq:hcmod4}} and~\stat{\eqref{eq:hcmod2.1}}
are obtained using properties of the $\SM$ operator and can be found in~\ref{sec:formalization}.
%
%
The formalization of Statement~\stat{\eqref{eq:hcmod2.1}} follows from the general result below, if we just replace predicate names $p$ and $q$ by $\mathit{in}$ and~$\mathit{r}$, respectively.
Its proof can also be found in~\ref{sec:formalization}.

\begin{proposition}
\label{prop:transitive.closure}
Let 
formula $F_{tr}^{qp}$ be 
\beq
\forall xy\big(p(x,y)\rar q(x,y)\big)\,\wedge
\forall xyz\big(q(x,z)\wedge q(z,y)\rar q(x,y)\big)
\eeq{eq:formulatheorem}
For any arbitrary predicates $p$ and $q$, any model $\cI$ of \definition 
$(q:F_{tr}^{qp})$
is such that the    
the extent of {\it q} is the transitive closure of the relation constructed from the extent of~{\it p}.
\end{proposition}

%
%
%
The next two results prove the correctness of~\stat{cn} and~\stat{hc}, respectively, for~$\mp_{cn}$ and~$\mh$. 

\begin{proposition}\label{prop:connectivity}
Let 
$F^{v}$ be formula
$\forall x y (  \neg{\it q}(x, y) \wedge {\it v}(x)\wedge 
{\it v}(y)\rar \bot)$; $F_{tr}^{qp}$ be formula~\eqref{eq:formulatheorem};
 $\MP{cn}^{vpq}$  
be a modular program $\tuple{\set{v,p}, \set{\,(q :F_{tr}^{qp}),\, F^v\,}}$,
$\cI$ be an interpretation,
and
$\tuple{v^{\cI},p^{\cI}}$
be a graph.
Then, $\cI$ is a model of~$\MP{cn}^{vpq}$
iff 
for every pair \mbox{$a,b \in v^{\cI}$} of distinct vertices,
there is a directed path from~$a$ to~$b$ in $\tuple{v^{\cI},p^{\cI}}$.
\end{proposition}
\begin{proofs}
We just show the left to right direction.
The complete proof can be found in~\ref{sec:formalization}.
%
Note first that $F^{v}$ is classically equivalent to  $\forall x y ( {\it v}(x)\wedge {\it v}(y)\rar {\it q}(x, y))$,
so any pair of vertices satisfies $(a,b) \in q^\cI$.
From Proposition~\ref{prop:transitive.closure},
relation~$q^\cI$ is the transitive closure of~$p^\cI$.
Hence,  path from $a$ to $b$  exists in graph $\tuple{v^{\cI},p^{\cI}}$.
\end{proofs}

\begin{proposition}\label{prop:connectivityb}
Let $\MP{cn}^{vpq}$ be as in Proposition~\ref{prop:connectivity}, $F^{p}$ be formula $\forall x y z({\it p}(x,y)\wedge {\it p}(x,z)\wedge \neg(y=z) \rar\bot)$, 
$\mh^{vpq}$ be modular program~$\tuple{\set{v,p}, \set{\, \MP{cn}^{vpq}, \, F^p \,}}$ and
$\cI$ be an interpretation of $\mh^{vpq}$.
Then, $\cI$ is a model of graph~$\mh^{vpq}$
iff $p^{\cI}$ are the edges of a Hamiltonian cycle of~$\tuple{v^{\cI},p^{\cI}}$
 that is,
    the elements of $\mathit{p}^\cI$ can be arranged as a directed cycle
    $(v_1,v_2),(v_2,v_3),\dots, (v_n,v_1)$
    so that $v_1,\dotsc,v_n$ are pairwise distinct
    and~$\mathit{v}^\cI=\{v_1,\dotsc,v_n\}$.
\end{proposition}
\begin{proofs}
Again, we show only the left to right direction.
See~\ref{sec:formalization} for the complete proof.
If $\cI$ is a model of $\mh^{vpq}$,
the hypothesis in the enunciate implies
that for any pair \mbox{$v_1,v_m \in v^{\cI}$} of vertices,
there is a directed path
$(v_1,v_2),\allowbreak(v_2,v_3),\allowbreak(v_3,v_4),\dots,\allowbreak(v_{m-1}, v_m)$
in  $\tuple{v^{\cI},p^{\cI}}$.
Hence,
there exists:
\beq 
(v_1,v_2),\allowbreak(v_2,v_3),\allowbreak(v_3,v_4),\dots,\allowbreak(v_{m-1},v_m),\allowbreak (v_{m+1},v_{m+1}),\dotsc,\allowbreak(v_n,v_1)
\eeq{eq:path}
in graph $\tuple{v^{\cI},p^{\cI}}$
such that  every vertex in $v^{\cI}$ appears in it.
Since $\cI$ is also a model of~$F^p$, Statement~\stat{\eqref{eq:hcmod2.1}} (modulo names of predicate symbols) is applicable.
This implies $v_i \neq v_j$ for all $i \neq j$
and, thus, that all edges in~\eqref{eq:path} are distinct.
Therefore, \eqref{eq:path} is a directed cycle.
Since this cycle  covers all vertices of~$\tuple{v^{\cI},p^{\cI}}$,
 it is also a Hamiltonian cycle.
\end{proofs}

\noindent
Modular programs~$\mp_{cn}$ and~$\mh$ coincide with modular programs
$\MP{cn}^{vpq}$ and~$\mh^{vpq}$,
when  predicate symbols~$v$,~$p$ and~$q$ are replaced by~$\mathit{vertex}$, $\mathit{in}$ and~$\mathit{r}$, respectively.
Statements~\stat{cn} and~\stat{hc} follow immediately.
Note also that modular programs~$\mp_{cn}$ and~$\mh$ also preserve their meaning inside a larger modular program mentioning predicate symbol~$r$ in other parts. Indeed, symbol $r$ is hidden or \emph{local} (existentially quantified) so that its use elsewhere in a larger modular program has a different meaning.

To complete the modularization of the \hc problem, we develop  the encoding for a graph instance.
Given a set $E$ of graph edges,  $M_E$ denotes a \definition
\beq
\big(\, edge :  \bigwedge\set{\mathit{edge}(a,b) \mid (a,b) \in E} \,\big)
\eeq{eq:pe}
%
The intuitive and formal meaning of \definition  $M_E$  is captured by the statement:
\begin{enumerate}
\setcounter{enumi}{7}
\item[\stat{E}:]\label{l:hcmod9}
``In any model of~$M_E$,
the extent of ${\it edge}$ is $E$.''
\end{enumerate}
Now, the \hc problem on a given graph instance with edges~$E$
is encoded by~$\MP{1}(E) := \tuple{ \, \set{\it in}, \, \set{ \MP{1} , \, M_E  } \, }$.
To prove that $\MP{1}(E)$ obtains the correct solutions, we can now just simply rely on the already proved fulfillment of the  statements for $\MP{1}$.

\begin{proposition}\label{prop:correct1}
Let $G = \tuple{V,E}$ be a graph with non-empty sets of vertices~$V$ and edges~$E$, where every vertex occurs in some edge, and~$\cI$ be an interpretation over signature~$\set{\mathit{in}}$.
Then,
$\cI$ is an answer set of~$\MP{1}(E)$ iff
$\mathit{in}^\cI$ is a \hc of $G$.
\end{proposition}

\begin{proofs}
We only showcase the left to right direction (see Appendix~B for the rest).
Take interpretation $\cI$ to be an answer set of~$\MP{1}(E)$.
Then, there exists a Herbrand interpretation $\cJ$ over signature $\inp,\vertex,\edge$ so that $\cJ$  coincides with~$\cI$ on the extent of $in$ and $\cJ$ is also a model of all submodules of $\MP{1}(E)$.
%
Thus, $\cJ$  adheres to  Statements~\stat{sg}, \stat{hc} and~\stat{E} about these submodules.
This implies that the extent~$\inp^\cI$ forms a \hc of~$\tuple{\vertex^{\cJ},\inp^{\cJ}}$ (Statement~\stat{hc}) and that $\tuple{\vertex^{\cJ},\inp^{\cJ}}$ is a subgraph of~$\tuple{\vertex^{\cJ},\edge^{\cJ}}$ (Statement~\stat{sg}).
These two facts imply that~$\inp^\cI$ forms a \hc of the graph~$\tuple{\vertex^{\cJ},\edge^{\cJ}}$.
Moreover, $\vertex^{\cJ}$ and~$\edge^{\cJ}$ respectively coincide with sets~$V$ and~$E$ (Statements~\stat{2} and~\stat{E}).
Therefore,~$\inp^\cJ$ forms a \hc of the graph~$\tuple{V,E} = G$.
Finally, recall that $\inp^{\cJ}=\inp^{\cI}$, so the result holds.
\end{proofs}
This result confirms that $\MP{1}$ is indeed a correct \emph{formal specification} of the \hc problem, for any arbitrary graph instance $E$.
Propositions~\ref{prop:transitive.closure}-\ref{prop:correct1} are, in fact, an example of the application of~\ref{item:methodology.3} of our methodology to the \hc problem.
It is worth to mention that, in most cases, the decomposition in modules and the properties of the~$\SM$ operator allow us to replace the~$\SM$ operator by a FO formula (using Clark's completion) or by the circumscription operator.
This replacement greatly simplify the effort of the proofs detailed in~\ref{sec:formalization}.

\section{Verification based on modular programs}
\label{sec:verif}

The results in the previous section state that the answer sets of our modular specification~$\MP{1}$  correspond to the \hcs of a graph.
However, in general, there is no guarantee that the non-modular version of $\Pi$ (i.e., the regular ASP program $P$ formed by all rules in $\mp$) has the same answer sets.
%
%
Next, we introduce  some general conditions under which the answer sets of a modular program $\mp$ and its non-modular version $P$ coincide.
These results are useful for~\ref{item:methodology.4} of the proposed methodology.

In the rest of the section, we assume that $\mp$ has the form $\tuple{\cS,\cM}$. 
We also identify a regular program $P$ with its direct modular version $\tuple{pred(P), \set{(pred(P):P)}}$.
%
The \emph{flattening} of $\mp$ is defined as $\flatt{\mp} \eqdef \tuple{\cS,\modules(\mp)}$.
For example, 
$\flatt{\MP{1}(E)} = \tuple{\set{\it in},\set{ \M{1},\allowbreak \M{2},\allowbreak \M{3},\allowbreak \eqref{eq:hcmod4} ,\allowbreak \eqref{eq:hcmod2.1},\allowbreak \eqref{eq:pe} }}$.
%
We say that $\mp$ is in $\alpha$-\emph{normal form} (\ANF) if all occurrences of a predicate name in~$\Phi(\mp)$ are free or they are all bound to a unique occurrence of existential quantifier. 
This happens, for instance, in~$\MP{1}(E)$.
Still, any formula $\Phi(\mp)$ can always be equivalently reduced to $\alpha$-NF by applying so-called $\alpha$-transformations (i.e., choosing new names for quantified variables).
In our context, this means changing hidden (auxiliary) predicate names in~$P$ until  \emph{different} symbols are used for distinct auxiliary predicates.
%
%
The next theorem states that, when modular program $\mp$ is in \ANF{}, we can ignore its recursive structure and instead consider its flat version.

\begin{proposition}\label{thm:flatten}
For any modular program~$\mp$ in \ANF{}, an interpretation  $\cI$ is a model of~$\mp$ iff $\cI$ is a model of $\flatt{\mp}$.
\end{proposition}
%
%
%
%


Thus, since~$\MP{1}(E)$ is in \ANF{}, it is simply equivalent to~$\flatt{\MP{1}(E)}$.
Next, we show  how to relate a flat program with a non-modular program that contains exactly the same rules. 
We use this result to verify that program $\fh$ (corresponding to Listing~\ref{list:hamiltonian}) satisfies the \hc specification.
To formalize this relation we  focus on a syntax closer to one of logic programs. Consider FO formulas that are conjunctions of rules of the form:
\beq
\ba{r}
\widetilde{\forall} (a_{k+1} \land \dots \land a_l \land \neg a_{l+1} \land \dots 
\land \neg a_m \,\land \neg \neg a_{m+1} \land \dots \land \neg \neg a_n \rar
a_1 
\lor \dots \lor a_k), 
\ea
\eeq{eq:rule_fo}
where all $a_i$ are atoms and $\widetilde{\forall}$ stands for the universal closure of their variables.
As usual, the consequent and the antecedent of~\eqref{eq:rule_fo} are  called {\em head} and  {\em body}, respectively.
The conjunction $a_{k+1} \land \dots \land a_{l}$ constitutes the {\em positive (part of the) body}.
A modular program~$\mp$ is called \emph{simple} if, for every \definition 
$(\vec{p} : F) \in \modules(\mp)$, formula~$F$ is a conjunction of rules, and all head predicate symbols in $F$ are intensional (occur in~$\bf p$).
This is, in fact, the case of all \definitions we considered so far.
Let $\intens(\mp)$ collect all intensional predicates in $\mp$, that is $\intens(\mp) \eqdef \bigcup \{ \vec{p} \mid (\vec{p} : F) \in \modules(P)\}$.
Then, we form the directed \emph{dependency graph} $DG[\mp]$ by taking $\intens(\mp)$ as nodes and adding an edge $(p,q)$ each time there is a rule occurring in $\mp$ with $p$ in the head and $q$ is in the positive body. 
For instance, the dependency graph of program~$\MP{1}(E)$ is given in Figure~\ref{fig:example_dep}.
This graph has four strongly connected components, each one consisting of a single node.

%
A modular program~$\mp$ is 
{\em coherent} if it is simple, in \ANF{}, and satisfies two more conditions: 
(i)  every pair of  distinct \definitions
 $(\vec{p} : F)$
 and 
\mbox{$(\vec{p'} : F')$}  in $\modules(\mp)$ is such that \mbox{$\vec{p} \cap \vec{p'} = \emptyset$}, 
and (ii)
 for every strongly connected component ${\tt SCC}$ in $DG[\mp]$, 
there is a \definition \hbox{$(\vec{p} : F) \in \modules(\mp)$} such that $\vec{p}$ contains 
all vertices in~${\tt SCC}$.
For example, $\MP{1}(E)$ is a coherent program.
Now, let us collect the conjunction of all \definition formulas in $\mp$ as
$\formula(\mp) \eqdef \bigwedge \{F \mid (\vec{p} : F) \in \modules(\mp)\}.$
We can observe, for instance, that
$\fh \wedge \formula(M_E) = \formula(\MP{1}(E)) \wedge  \eqref{eq:hcmod2.2}  = \formula(\flatt{\MP{1}(E)}) \wedge  \eqref{eq:hcmod2.2} $, that is, the modular encoding of the \hc problem and the non-modular one share the same rules but for~\eqref{eq:hcmod2.2}.
We now obtain a similar result to Theorem~3 in~\citep{har16} but
extended to our framework.
Together with Proposition~\ref{thm:sm-traditional}, it connects modular programs with logic programs as used in practice.
The proof of this result together with the proof of Theorem~2 below can be found in~\ref{sec:proofs}.

\begin{theorem}\label{thm:sm-traditional2.hidden}
Let $\mp = \tuple{\cS,\cM}$ be a coherent modular program, $\vec{p}$ be  $\intens(\mp)$, and $\vec{h}$ be $pred(\Phi(\mp))\setminus \cS$.
Then, (i) any interpretation~$\cI$ is a model of $\mp$ iff $\cI$ is a model of formula $\exists \bf h$ $\SM_\vec{p}[\formula(\mp)]$; (ii)
any  interpretation~$\cI$ is an answer set of $\mp$ iff
 there is some answer set $\cJ$ of~$\formula(\mp)$ such that $\cI=\restr{\cJ}{\cS}$.
\end{theorem}

As a result, we can now prove that program in Listing~\ref{list:hamiltonian}  satisfies the formal specification~$\MP{1}$ which, as we saw, captures the \hc problem.

\begin{proposition}
The answer sets of modular program~$\MP{1}(E)$
coincide with the answer sets of logic program $\fh \wedge\formula(M_{E})$ for intensional predicate symbol `{\rm \it in}' hiding all other predicate symbols of the program.

\end{proposition}
\begin{proof}
Since modular program~$\MP{1}(E)$ is coherent, by Theorem~\ref{thm:sm-traditional2.hidden}, it is equivalent to the formula
to~\mbox{$\varphi := \exists\vec{h} \; \text{SM}_\vec{p}[\formula(\MP{1}) \wedge \formula(M_E)]$} where~\mbox{$\vec{h}=\tuple{vertex,edge,r}$}.
Now, $\varphi$ is in its turn equivalent 
to \mbox{$\exists\vec{h} \text{SM}_\vec{p}[\formula(\MP{1}) \wedge \eqref{eq:hcmod2.2} \wedge \formula(M_E) ]$}
since \mbox{$\formula(\MP{1})$} entails formula~\eqref{eq:hcmod2.2} and~\definition \eqref{eq:hcmod2.2} has no intensional predicate symbols.
Besides, formulas \mbox{$\formula(\MP{1}) \wedge \eqref{eq:hcmod2.2}$} and $\fh$ are identical.
From Proposition~\ref{thm:sm-traditional}, it follows that
 Herbrand models of~$\varphi$ are the answer sets of $\fh \wedge \formula(M_E)$ restricted to predicate~$in$.
\end{proof}

%
This proposition constitutes  verification \ref{item:methodology.4} that links an ASP encoding~$\fh$ of the \hc problem to its formal specification as a modular program~$\MP{1}$.
%
At a first sight, the effort may seem worthless, given that $\MP{1}$ and $\fh$ almost share the same rules.
But this is a wrong impression, since $\fh$ is actually an ideal case, i.e. the one closest to $\MP{1}$, while the latter can still be used as a specification for other encodings.
To show how, let us take another encoding $\fh'$ that results from replacing~\eqref{eq:hcr1}-\eqref{eq:hcmod4} in $\fh$ by rules in Listing~\ref{list:hamiltonian2} respectively corresponding to:
\begin{align}
&\forall y\big(\inp(a,y) \rar \rap(y)\big)
\label{eq:hcmod1.1a}
\\
&\forall x y \big(\inp(x,y) \wedge \rap(x)\rar \rap(y)\big) 
\label{eq:hcmod3.2b}
\\
&\forall y \big(\neg {\it ra}(y)\wedge {\it vertex}(y) \rar \bot\big)
\label{eq:hcmod4.2a}
\end{align}
Verifying program $\fh'$ amounts to proving its adherence to~$\MP{1}$ and, for that purpose, requires a proper modularization $\MP{1}'$ of~$\fh'$.
In this case, that modularization is obvious since the change is \emph{local} to the module checking \hcs, $\MP{hc}$.
We define the modular programs $\MP{cn}' := \tuple{\set{\it vertex, in},\set{ (ra: \eqref{eq:hcmod1.1a}\wedge\eqref{eq:hcmod3.2b}),\eqref{eq:hcmod4.2a} }}$,  $\MP{hc}' := \tuple{\set{\it vertex, in},\set{ \MP{cn}', \eqref{eq:hcmod2.1}, \eqref{eq:hcmod2.2} }}$ and $\MP{1}'$, as the result of replacing $\MP{hc}$ by $\MP{hc}'$ in $\MP{1}$.
Even though they use different auxiliary predicates, programs $\MP{hc}$ and $\MP{hc}'$ have the same intuitive meaning (Statement \stat{hc}) as long as there exists some vertex $a$ in the graph. 
One would, therefore, expect that the correctness of $\MP{hc}'$ could be proved by checking some kind of equivalence with respect to $\MP{hc}$.
We formalize next this idea.

Given modular programs~$\mp,\mp_1$ and $\mp_2$, we write $\mp[\mp_1/\mp_2]$ to denote the result of replacing all occurrences of module~$\mp_1$ in~$\mp$ by $\mp_2$.
We also define $\Phi(\mp-\mp_1) \eqdef \bigwedge \{\Phi(M) \mid M \in \cM, \mp_1 \not\in \modules(M) \} $.
For any finite theory $\Gamma$, two modular programs~$\mp_1$ and~$\mp_2$ are said to be \emph{strongly equivalent with respect to context~$\Gamma$} when any modular program~$\mp$ with $\Phi(\mp - \mp_1) \models \Gamma$
satisfies that $\mp$ and~$\mp[\mp_1/\mp_2]$ have the same answer sets.
%

%
%
%

\begin{theorem}
\label{thm:modular.equivalence}
Two modular programs~$\mp$ and~$\mp'$ are strongly equivalent under context~$\Gamma$ iff
\mbox{$\Gamma \models \Phi(\mp) \leftrightarrow \Phi(\mp')$} holds for all Herbrand interpretations.
\end{theorem}

In our example, although $\MP{hc}$ and $\MP{hc}'$ are not equivalent in general, we can prove:

\begin{proposition}\label{prop:ham.equivalence}
Modules $\MP{hc}$ and~$\MP{hc}'$ are strongly equivalent w.r.t.~$\Gamma=\set{ \vertex(a) }$.
\end{proposition}
\begin{proofs}
Recall that
\mbox{$\mp_{\it hc} \models\eqref{eq:hcmod2.2}$}.
The rest of the proof follows two steps.
First, given:
\begin{align}
&\exists r\, \big(\ \Phi(\M{3}) \wedge \forall y (\vertex(a)\wedge\vertex(y)  \rar \rp(a,y)) \ \big)
\label{eq:hcmod4.4}
\end{align}
we get 
\mbox{$\vertex(a) \models \Phi(\MP{cn}') \leftrightarrow \eqref{eq:hcmod4.4} $}
and, furthermore, 
$\models \Phi(\mp_{\it c}) \to \eqref{eq:hcmod4.4}$
follows by instantiation of~$\forall x$ with \mbox{$x=a$}.
Second, we can prove
\mbox{$\eqref{eq:hcmod2.1} \models \eqref{eq:hcmod4.4} \to \Phi(\mp_{\it c})$}.
\end{proofs}



\section{Conclusions and future work} \label{sec:conc}
We  presented a modular ASP framework that allows nested modules possibly containing hidden local predicates.
The semantics of these programs and their modules is specified via  the second-order SM operator, and so, it does not resort to grounding.
We  illustrated how, under some reasonable conditions, a modular reorganization of a logic program can be used for verifying that it adheres to its (in)formal specification.
This method has two important advantages.
First, it applies a divide-and-conquer strategy, decomposing the correctness proof for the target program into almost self-evident pieces.
Second, it can be used to guarantee correctness of a module replacement, even if interchanged modules are non-ground and use different local predicates.
In this way, correctness proofs are also reusable.
%
The need for second-order logic is inherent to the expressiveness of first-order stable models but has the disadvantage of lacking a proof theory in the general case. Yet, 
there are well-known specific cases in which the second-order quantifiers can be removed.
This is often the case of the SM operator
so that we can use formal results from the literature (splitting, head-cycle free transformations, relation to Clark's Completion or Circumscription -- see \citep{feleli11a})  to reduce the these second-order formulas to first-order ones.
We also intend to exploit the correspondence between SM and Equilibrium Logic to study  general inter-theory relations~\cite{PV04}.

Our definition of contextual strong equivalence using hidden predicates is a variation of \emph{strong equivalence}~\cite{lipeva01a,lif07a}.
We leave it to future work the relation to other program equivalence and correspondence notions~\cite{eitowo05a,oettom08,OikarinenJ09,AguadoCFPPV19,GeibingerT19}.
%
%
%
%
Another topic for future work is the extension of automated reasoning tools for ASP verification~\cite{lifschitz18a} to incorporate modularity. 

\paragraph{\bf Acknowledgments.}
We are thankful to Vladimir Lifschitz and the anonymous reviewers for their comments that help us to improve the paper.
This work was partially supported by MINECO, Spain, grant TIC2017-84453-P and NSF, USA grant 1707371.
The second author is funded by the Alexander von Humboldt Foundation.

\bibliographystyle{acmtrans}%
\bibliography{krr,bib,procs}%
\clearpage
\appendix
\section{Formalizing the meanings of \definitions for\\ the \hc problem}
\label{sec:formalization}

In this section we argue about the correctness of statements introduced in Section~\ref{sec:specification}.
We start by reviewing a series of results that are useful in proving such correctness and apply them to our running example.

In this section, it is convenient for us to  identify  {\definition} \mbox{$M=(\vec{p}:F)$} with the formula $\Phi(M) \eqdef \SM_\vec{p}[F]$ that captures the semantics of $M$.

\paragraph{Strong equivalence and denials.}
Non-modular programs $P_1$ and~$P_2$ are {\em strongly equivalent} if for 
every traditional program $P$, programs $P_1\cup P$ and $P_2\cup P$ have the same answer sets~\cite{lipeva01a}.
More in general,
FO formulas $F$ and $G$ are {\em strongly equivalent} if 
for any formula $H$, any occurrence of~$F$ in~$H$, and any list~{\bf p} of 
distinct predicate constants,  SM$_{\bf p}[H]$ is equivalent to 
SM$_{\bf p}[H']$, where~$H'$ is obtained from $H$ by replacing 
$F$ by $G$.
\citeN{lif07a} show that FO 
formulas~$F$ and~$G$ are strongly equivalent if they are equivalent in \sqht logic\,---\,a FO intermediate logic~\cite{PearceV08} between classical and intuitionistic logics.
%

A formula of the form~${\widetilde{\forall}(\body\rar \bot)}$ 
is intuitionistically equivalent to formula $\neg \widetilde{\exists} \body$.
We call formulas of both of these forms {\em denials} and identify the former with the latter.
%
\begin{theorem}[Theorem 3;~\citeNP{feleli11a}]\label{thm:constraints}
	For any FO formulas $F$ and $G$ and arbitrary tuple~${\bf p}$ of predicate constants, $\SM_{\bf p}[F\wedge\neg G]$ is equivalent to 
	$\SM_{\bf p}[F]\wedge\neg G.$
\end{theorem}
Theorem~\ref{thm:constraints} can be understood in the following terms. 
A modular program $\Pi$ containing a \definition of the form \mbox{$(\vec{p}:F\wedge \neg G)$}
is equivalent to the one resulting from $\Pi$ by replacing \mbox{$(\vec{p}:F\wedge \neg G)$} with \definitions  
 \mbox{$(\vec{p}:F)$} and $\neg G$. Thus, any denial semantically translates into a classical first order formula.
Now, \textit{the claims in Statement~\stat{\eqref{eq:hcmod4}} and~\stat{\eqref{eq:hcmod2.1}} immediately follow from Theorem~\ref{thm:constraints}.}
%

\paragraph{Tightness and completion.}
Although $\SM$ is defined on arbitrary formulas, we focus now on the traditional syntax of logic program rules, that is, FO sentences of the form~\eqref{eq:rule_fo}.
We say that a module is {\em tight} if its dependency graph  is acyclic.
For example, all modules in program $\Pi_1$, but~\MP{cn} (and those containing it) are tight.

A FO formula $F$ is in {\em  Clark normal form}~\cite{feleli11a} relative to the tuple~{\bf p} of  predicate symbols if it is a conjunction of formulas of the form
\beq 
\forall \vec{x} (G\rar p(\vec{x}))
\eeq{eq:compformula}
one for each predicate $p\in {\bf p}$, where $\vec{x}$ is a tuple of distinct object variables. 
We refer the reader to Section 6.1 in~\cite{feleli11a} for the description of  the intuitionistically equivalent transformations that can convert a FO formula that is a conjunction of formulas of the form~\eqref{eq:rule_fo}   into Clark normal form. 
Here, we illustrate results of these conversion steps on formulas stemming from the program $\Pi_1$.
For instance, converting 
formula
\beq 
{\it edge}(a, a')\wedge \;\; \dots \;\; \wedge {\it 
edge}(c, c')
\eeq{eq:edgeformula}
into Clark normal form 
results in the intuitionistically equivalent formula
\begingroup
\setlength{\thinmuskip}{0mu}
\setlength{\medmuskip}{1mu}
\setlength{\thickmuskip}{1mu}
\beq
\forall xy ((x=a\wedge y=a')\vee \dots
\vee
(x=c\wedge y=c'))\rar edge(x,y))
\eeq{eq:edgeclark}
\endgroup
Similarly, module~$\Pi_E$
is intuitionistically equivalent to the formula 
\beq{\forall z(F_{\mathit{edge}}(z)  \rar {\it vertex}(z))},\eeq{eq:fedge}
where $F_{\mathit{edge}}(z)$ follows
\begin{gather*}
\exists  xy  ({\it edge}(x,y)\wedge z=x) \vee \exists xy( {\it edge}(y,x) \wedge z=x).
\end{gather*}
The FO formula within \definition~\eqref{eq:choicerulefo} is in Clark normal form.

The {\em completion} of a formula  $F$ in Clark normal form relative to predicate symbols~{\bf p}, denoted by $\COMP_{\bf p}[F]$, 
is obtained from $F$ by replacing each conjunctive term of the form~\eqref{eq:compformula}
by 
$$
\forall \vec{x} (G\lrar p(\vec{x})).
$$
For instance, the completion of~\eqref{eq:edgeclark} is
\begingroup
\setlength{\thinmuskip}{0mu}
\setlength{\medmuskip}{1mu}
\setlength{\thickmuskip}{1mu}
\beq
\forall xy ((x=a\wedge y=a')\vee \dots
\vee
(x=c\wedge y=c'))\lrar edge(x,y)),
\eeq{eq:edgecomp}
\endgroup
while the completion of formula~\eqref{eq:fedge}
is 
\beq{\forall z(F_{\mathit{edge}}(z)  \lrar {\it vertex}(z))}.
\eeq{eq:vertexcomp}

The following theorem follows immediately from Theorem 11 in~\cite{feleli11a}.
\begin{theorem}\label{thm:completion}
Let $\SM_{\bf p}[F]$ be a tight \definition.
Then, $\SM_{\bf p}[F]$ and~$\COMP[F]$ are equivalent.
\end{theorem}
Note that expression $\COMP[F]$ is a classical first order logic formula.
Since formulas~\eqref{eq:edgeformula} and~\eqref{eq:edgeclark} are strongly equivalent, it follows that \definitions~$\Pi_E$ and~$\SM_{edge}[\eqref{eq:edgeclark}]$ are equivalent, too.
Similarly, \definition~$M_1$ is equivalent to SO formula $\SM_{vertex}[\eqref{eq:fedge}]$.
Thus, by Theorem~\ref{thm:completion}, modules~$M_E$
and~$M_1$ are equivalent to FO formulas~\eqref{eq:edgecomp}
 and~\eqref{eq:vertexcomp}, respectively.
\emph{These facts suffice to support the claims of Statements~\stat{E}
and~\stat{1}}.
Furthermore, also by Theorem~\ref{thm:completion}, \definition~$M_2$ is equivalent to FO formula
$$
\forall x y( (\neg\neg  {\it in}(x,y) \wedge {\it edge}(x,y)) \lrar {\it in}(x,y)),
$$
which, in turn, is equivalent to formula
$$
\forall x y({\it in}(x,y) \rar {\it edge}(x,y)).
$$
\emph{It is easy to see now that the  claim of Statement~\stat{2} holds.}

The following proposition follows immediately from Theorem~\ref{thm:completion} and generalizes the last claim.
\begin{prop}\label{thm:completionchoice}
A tight \definition $SM_{p}[\forall \vec{x} (\neg\neg p(\vec{x})\wedge G\rar p(\vec{x}))]$ is equivalent to
formula
$\forall \vec{x} ( p(\vec{x})\rar G)$.
\end{prop}
In other words, we can always understand a \definition consisting of a choice rule as a reversed implication in FO logic.

\paragraph{Circumscription and transitive closure.}
The \emph{circumscription operator with the minimized predicates}~$\bf p$ of a FO formula $F$
is denoted by $\CIRC_{\bf p}[F]$~\cite{feleli11a}.
The models of $\CIRC_{\bf p}[F]$ are the models of~$F$ where the extension of the predicates in~${\bf p}$ is minimal given the interpretation of remaining predicates is fixed.
Interestingly,
if $F$ is a conjunction of rules of~\eqref{eq:rule_fo} without negation,
then $\SM_{\bf p}[F]$ and $\CIRC_{\bf p}[F]$ are equivalent.
Proposition~\ref{prop:transitive.closure} follows directly from this observation and allows us to assert that Statement~\stat{3} holds.

Let us now address the proofs of Propositions~\ref{prop:connectivity}-\ref{prop:correct1}.

For a formula $F$ and a symbol $q$ occurring in it by $(F)^q_Q$ we denote an expression constructed from $F$ by substituting symbol $q$  with $Q$.

\begin{proof}[\bf Proof of Proposition~\ref{prop:connectivity}]
Note that
\[
\Phi(\MP{cn}^{vpq})
	\ \ = \ \ \exists q \ \big(\Phi(q :F_{tr}^{qp}) \wedge \Phi(F^v))
	\ \ = \ \ \exists q \ \big(\Phi(q :F_{tr}^{qp}) \wedge F^v)
\]
Let $\cI$ be an interpretation over signature~$p,v$.
Then, $\cI$ is a model of~$\MP{cn}^{vpq}$
iff there is an interpretation~$\cJ$ over signature~$p,q,v$ that agrees on~$\cI$ on~$p,v$ and is a model of~$\big(\Phi(q :F_{tr}^{qp}) \wedge F^v)$.
Take such interpretation $\cJ$, its $q^J$ is the transitive closure of~$p^\cI$ (Proposition~\ref{prop:transitive.closure})
and~$\cJ$ is a model of~$F^v$.
%
Note also that $F^{v}$ is equivalent to
\begin{gather}
    \forall x y ( {\it v}(x)\wedge {\it v}(y)\rar {\it q}(x, y)).
    \label{eq:1:proof:prop:connectivity}
\end{gather}
Consequently, any pair of elements in $v^\cI$ are such that they are also in binary relation~$q^\cJ$ with each other.
\emph{Left-to-right.}
Assume that~$\cI$ is a model of~$\MP{cn}^{vpq}$.
Then, any pair of vertices satisfies ${(a,b) \in q^\cJ}$
and, consequently, a path from $a$ to $b$ exists in graph~$\tuple{v^{\cI},p^{\cI}}$.
\mbox{\emph{Right-to-left.}}
Let~$\cI$ be an interpretation of~$\it v$ and~$\it p$ such that there is  a direct path exists connecting all pairs of vertices in~$\tuple{\it v^\cI, p^\cI}$.
Let~$q^\cJ$ be the transitive closure of~$p^\cI$.
From Proposition~\ref{prop:transitive.closure}, this implies that $\cJ$ satisfies~$\Phi(q :F_{tr}^{qp})$.
Take now any two vertices~${a,b\in \it v^\cI}$.
Then, there is a path~$(v_0,v_1),(v_1,v_2), \dotsc, (v_{n-1}, v_n)$ such that $v_0=a$, $v_n=b$ and $(v_{i-1},v_i) \in \it p^\cI$ for all $1 \leq i \leq n$.
Since~$q^\cJ$ is the transitive closure of~$p^\cI$, it follows that $(a,b) \in q^\cJ$ and, thus,~$\cI$ satisfies~$\eqref{eq:1:proof:prop:connectivity}$.
Consequently, $\cI$ is a model of~$\MP{cn}^{vpq}$.
\end{proof}

\begin{proof}[\bf Proof of Proposition~\ref{prop:connectivityb}]
\emph{Left-to-right.}
If $\cI$ is a model of $\mh^{vpq}$, then Proposition~\ref{prop:connectivity} implies that for any pair \mbox{$v_1,v_m \in v^{\cI}$} of vertices, there are directed paths
$(v_1,v_2),\allowbreak(v_2,v_3),\allowbreak(v_3,v_4),\dots,\allowbreak(v_{m-1}, v_m)$
and
$(v_{m+1},v_{m+1}),\dotsc,\allowbreak(v_n,v_1)$
in $\tuple{v^{\cI},p^{\cI}}$.
Hence, there exists:
\begin{gather}
(v_1,v_2),\allowbreak(v_2,v_3),\allowbreak(v_3,v_4),\dots,\allowbreak(v_{m-1},v_m),\allowbreak (v_{m+1},v_{m+1}),\dotsc,\allowbreak(v_n,v_1)
\tag{\ref{eq:path}}
\end{gather}
in graph $\tuple{v^{\cI},p^{\cI}}$
such that  every vertex in $v^{\cI}$ appears in it.
Since $\cI$ is also a model of~$F^p$, Statement~\stat{\eqref{eq:hcmod2.1}} (modulo names of predicate symbols) is applicable.
This implies $v_i \neq v_j$ for all $i \neq j$
and, thus, that all edges in~\eqref{eq:path} are distinct.
Therefore, \eqref{eq:path} is a directed cycle.
Since this cycle  covers all vertices of~$\tuple{v^{\cI},p^{\cI}}$, it is also a Hamiltonian cycle.
\emph{Right-to-Left.}
If the elements of $\mathit{p}^\cI$ can be arranged as a directed cycle $(v_1,v_2),\allowbreak(v_2,v_3),\allowbreak\dots,\allowbreak (v_n,v_1)$ such that $\mathit{v}^\cI=\{v_1,\dotsc,v_n\}$, then it is clear that there is a path between each pair of vertices.
Consequently, $\cI$ is a model of~$\mh^{vpq}$ (Proposition~\ref{prop:connectivity}).
Hence, it only remains to check that~$\cI$ satisfies~$F^p$.
Suppose, for the sake of contradiction, that this is not the case.
Then, there are $(v_i,v_j) \in p^\cI$ and~$(v_i,v_k) \in p^\cI$ such that $v_j \neq v_k$,
which is a contradiction with the assumption that the elements of $\mathit{p}^\cI$ can be arranged as a directed cycle.
\end{proof}

\begin{proof}[\bf Proof of Proposition~\ref{prop:correct1}]
\emph{Left-to-right.}
Assume that interpretation $\cI$ is an answer set of~$\MP{1}(E)$.
Then, there exists an Herbrand interpretation $\cJ$ over signature $\inp,\vertex,\edge$ so that $\cJ$  coincides with~$\cI$ on the extent of $\inp$ and $\cJ$ is also a model of all submodules of $\MP{1}(E)$.
%
Thus, $\cJ$  adheres to  Statements~\stat{sg}, \stat{hc} and~\stat{E} about these submodules.
This implies that the extent~$\mathit{in}^\cI$ forms a \hc of~$\tuple{\vertex^{\cJ},\inp^{\cJ}}$ (Statement~\stat{hc}) and that $\tuple{\vertex^{\cJ},\inp^{\cJ}}$ is a subgraph of~$\tuple{\vertex^{\cJ},\edge^{\cJ}}$ (Statement~\stat{sg}).
These two facts imply that~$\mathit{in}^\cI$ forms a \hc of the graph~$\tuple{\vertex^{\cJ},\edge^{\cJ}}$.
Moreover, $\vertex^{\cJ}$ and~$\edge^{\cJ}$ respectively coincide with sets~$V$ and~$E$ (Statements~\stat{2} and~\stat{E}).
Therefore,~$\inp^\cJ$ forms a \hc of the graph~$\tuple{V,E} = G$.
Finally, recall that $\inp^{\cJ}=\inp^{\cI}$, so the result holds.
\textit{Right-to-left}.
Let $H$ be Hamiltonian cycle of a graph edges~$E$.
Let~$\cI$ be an interpretation over signature~$\inp$ such that~$\inp^\cI$ are all edges in~$H$
and let~$\cJ$ be an interpretation over signature $\inp,\vertex,\edge$ so that $\cJ$ coincides with~$\cI$ on the extent of $\inp$.
Assume also that~$\edge^\cJ$ are all edges in~$E$ and~$\vertex^\cJ$ contains exactly all objects occurring in the relation~$\edge^\cJ$.
Clearly~$\cJ$  adheres to Statements~\stat{sg}, \stat{hc} and~\stat{E} about these submodules and, therefore, $\cJ$ is a model of~$\Pi_1(E)$.
\end{proof}

\section{Proofs of Theorems~\ref{thm:sm-traditional2.hidden} and~\ref{thm:modular.equivalence}}
\label{sec:proofs}

\begin{lemma}\label{thm:sm-traditional2b}
For any coherent flat program $\mp$ with $\pred(\mp) = \vec{p}$
and~$\bf h$ the hidden predicates of~$\mp$, and
 any interpretation~$\cI$,
the following conditions are equivalent:
\setlength\leftmargin{12pt}
\begin{itemize}
\setlength\itemsep{0em}
\item $\cI$ is model of $\exists{\bf h}\,$\SM$_\vec{p}[\formula(\mp)]$,
\item $\cI$ is model $\mp$.   
\end{itemize}
\end{lemma}

\begin{proof}
Let $\mp = \tuple{\cS,\cM}$  with 
$\cM = \{({\bf p_1} : F_1 ), \dots , 
({\bf p_n} : F_n )\}$
be a flat modular program.
The proof is by induction on $n$. The base case is trivial. In the induction 
step, we assume that
for any coherent modular program with less than $n$ modules the result holds.
By definition it follows that
$$
\exists{\bf h}\,\text{SM}_{\pred(\mp)}\big[ \ \formula(\mp)
\ \big]
\quad = \quad
\exists{\bf h}\, \text{SM}_{\pred(\mp') \cup {\bf p_n}}\Big[ \ \bigwedge_{i = 1}^{n-1} F_i \;\; \land \;\; F_{n}
\ \Big],
$$
where 
$\mp' = \tuple{\cS',\cM \setminus \{ ({\bf p}_n : F_n )\} }$, where~$\cS'$ is the set containing all predicates $\mp'$.
By the Splitting Theorem in~\cite{felelipa09a}
the latter is equivalent to
\beq
\exists{\bf h}\, \Big( \text{SM}_{\pred(\mp')} \Big[ \  \bigwedge_{i = i}^{n-1} F_i \ \Big] \;\; \land \;\; 
\text{SM}_{\bf p_{n}} [F_{n}]  
\Big).
\eeq{eq:splitn+1}
By the definition of $\formula$,
$$
\text{SM}_{\pred(\mp')}\Big[ \  \bigwedge_{i = 1}^{n-1} F_i \ \Big]
\quad = \quad
\text{SM}_{\pred(\mp')}\big[ \  \formula(\mp') \ \big]
$$
and, by induction hypothesis, we get that
$\cI$ is a model of $\text{SM}_{\pred(\mp')}\big[ \  \formula(\mp') \ \big]$
iff
$\cI$ is a model of $\mp'$.
By definition, the latter holds iff
$\cI$ is a model of
$$
\bigwedge_{i=1}^{n-1} \text{SM}_{\vec{p_i}}[ F_i ].
$$
Hence, 
$$
\text{SM}_{\pred(\mp')}\Big[ \  \bigwedge_{i = 1}^{n-1} F_i \ \Big]
\qquad\text{and}\qquad
\bigwedge_{i=1}^{n-1} \text{SM}_{\vec{p_i}}[ F_i ]
$$
are equivalent.
This implies that~\eqref{eq:splitn+1}
is equivalent
\begin{gather}
\exists{\bf h}\, \left( \bigwedge_{i=1}^{n-1} \text{SM}_{\vec{p_i}}[ F_i ] \;\; \land \;\; 
\text{SM}_{\bf p_{n}} [F_{n}] \right),
    \label{eq:splitn+2}
\end{gather}
Finally, by definition, $\cI$ is a model of $\mp$ iff~$\cI$ is a model of~\eqref{eq:splitn+2}.
\end{proof}

\begin{proof}[\bf Proof of Theorem~\ref{thm:sm-traditional2.hidden}]
By definition, 
$\cI$ is model of $\mp=\tuple{\cS,\cM}$
iff $\cI$ is a model of
$$
\Phi(\mp)
	\ \ = \ \ \exists \vec{h} \
\bigwedge_{M \in \cM}
\Phi(M)
$$

Let us now define the degree of a modular program as follows
\begin{gather*}
\mathit{deg}(\mp) \quad \eqdef \quad 
       1 + \max\setm{ \mathit{deg}(M) }{ M \in \cM }
\end{gather*}
with $\mathit{deg}(M) = 0$ for every \definition~$M$.
We now construct a proof by induction on degree of a modular program. 
Let $n$ be a degree of modular program $\Pi$.
We state an induction hypothesis as follows.
For every modular program $\Pi'$ with $\mathit{deg}(\mp') < n$,
an interpretation
$\cI$ is a model of $\mp'$ iff $\cI$ is a model of $\exists\,\vec{h}\text{SM}_\vec{p}[\formula(\mp')]$.

The base case follows from Lemma~\ref{thm:sm-traditional2b}.

Let $\cM_{\it def} = \{ M \in \cM \mid M \text{ is a \definition} \}$ be the set of all \definitions in~$\cM$
and let
$\mp' = \tuple{\cS,\cM'}$
with 
$$
\cM' \quad\eqdef\quad \cM_{\it def} \cup \bigcup_{ \tuple{\cS',\cM''} \in \cM \setminus \cM_{\it def}} \cM''
$$
Obviously, if $\mp$ is not flat, then $\mathit{deg}(\mp') < \mathit{deg}(\mp)$
and, by  induction hypothesis, we obtain that
$\cI$ is model of $\mp'$ iff $\cI$ is a model of $\exists\vec{h}\,\text{SM}_\vec{p}[\formula(\mp')]$.
The latter holds iff $\cI$ is a model of $\exists\vec{h}\,\text{SM}_\vec{p}[\formula(\mp)]$
because
\begin{align*}
\formula(\mp)
    \quad&=\quad
    \bigwedge_{M \in \cM} \formula(M)
    \\
    \quad&=\quad
    \bigwedge_{M \in \cM_{\it def}} \formula(M)
        \wedge
        \bigwedge_{M \in \cM \setminus \cM_{\it def}} \formula(M)
    \\
    \quad&=\quad
    \bigwedge_{M \in \cM_{\it def}} \formula(M)
        \wedge
        \bigwedge_{M \in \cM \setminus \cM_{\it def}}
        \quad
        \Big( \
        \bigwedge_{\substack{M' \in \cM''\\M=\tuple{\cS',\cM''}}} \formula(M')
        \ \Big)
    \\
    \quad&\equiv\quad
    \bigwedge_{M \in \cM_{\it def}} \formula(M)
        \wedge
        \bigwedge_{\substack{\tuple{\cS',\cM''}\in  \cM \setminus \cM_{\it def}\\M' \in \cM''}} \formula(M')
    \\
    \quad&=\quad
    \bigwedge_{M \in \cM'} \formula(M)
    \\
    \quad&=\quad
    \formula(\mp').
\end{align*}
Let us show now that
$\cI$ is model of $\mp$ iff $\cI$ is model of $\mp'$.
We can see that, $\cI$ is a model of $\mp$ iff $\cI$ is a model of
\begin{align}
\Phi(\mp)
	\quad &= \quad
	\exists \vec{h} \ \bigwedge_{M \in \cM} \Phi(M)
    \\
    \quad&=\quad
    \exists \vec{h} \Big( \bigwedge_{M \in \cM_{\it def}} \Phi(M)
        \wedge
        \bigwedge_{M \in \cM \setminus \cM_{\it def}} \Phi(M) \Big)
    \\
    \quad&=\quad
    \exists \vec{h} \Big( \bigwedge_{M \in \cM_{\it def}} \Phi(M)
        \wedge
        \bigwedge_{M \in \cM \setminus \cM_{\it def}} \hspace{5pt} 
        \exists \vec{h}' \bigwedge_{\substack{M' \in \cM''\\M = \tuple{\cS',\cM''}}} \Phi(M')
        \Big)
        \label{eq:3:thm:sm-traditional2.hidden}
\end{align}
where $\vec{h}'$ is a tuple containing all predicates occurring in $\cM''$ that are not in $\cS'$.
Since the program is coherent,
then every predicate symbol occurring in two different modular programs or \definitions does not occur in $\vec{h'}$
and, thus, we obtain that~\eqref{eq:3:thm:sm-traditional2.hidden} is equivalent to
\begin{align}
\exists \vec{h}\exists \vec{h}'' \Big( \bigwedge_{M \in \cM_{\it def}} \Phi(M)
        \wedge
        \bigwedge_{\tuple{\cS',\cM''} \in \cM \setminus \cM_{\it def}} \hspace{5pt}  \bigwedge_{M \in \cM'' } \Phi(M)
        \Big)
        \label{eq:4:thm:sm-traditional2.hidden}
\end{align}
where $\vec{h}''$ is a tuple containing all predicates in $\vec{h}'$ for all $\mp' \in \cM \setminus \cM_{\it def}$.
Since the program is coherent,
we also get that every predicate symbol occurring in $\vec{h}''$ also occurs in $\vec{h}$,
and, thus,~\eqref{eq:4:thm:sm-traditional2.hidden}
is equivalent to
\begin{align*}
\quad
    \exists \vec{h} \Big( \bigwedge_{M \in \cM_{\it def}} \Phi(M)
        \wedge
        \bigwedge_{\tuple{\cS',\cM''} \in \cM \setminus \cM_{\it def}} \hspace{5pt}  \bigwedge_{M \in \cM''} \Phi(M)
        \Big),
\end{align*}
which in turn is $\Phi(\mp')$.
This concludes the argument that 
$\cI$ is a model of $\mp$
iff
$\cI$ is a model of $\mp'$.
By arguments at the beginning of the proof it follows that
$\cI$ is a model of $\mp'$ iff
$\cI$ is model of $\exists\vec{h}\,\text{SM}_\vec{p}[\formula(\mp')]$
iff
$\cI$ is model of $\exists\vec{h}\,\text{SM}_\vec{p}[\formula(\mp)]$.
This concludes the proof of (i) point of the theorem.  
The (ii) point follows from Proposition~\ref{thm:sm-traditional}.
\end{proof}


\begin{proof}[\bf Proof of Theorem~\ref{thm:modular.equivalence}]
The \emph{left-to-right} direction is immediate.
For the \emph{right-to-left} direction, assume that~$\Pi$ and~$\Pi'$ are strongly equivalent under context~$\Gamma$
and suppose, for the sake of contradiction, that there is a Herbrand  model~$\cI$ of~$\Gamma$
that does not satisfy~$\Phi(\mp) \leftrightarrow \Phi(\mp')$.
Assume, without loss of generality, that
$\cI$ is a model of~$\Gamma$ and~$\Phi(\mp)$ but not of~$\Phi(\mp')$.
Let~$F$ be the conjunction of all formulas in~$\Gamma$
and let~$\Pi_1 = \tuple{\cS,\cM_1}$ and $\Pi_2 = \tuple{\cS,\cM_2}$,
where~$\cS$ the set of all predicate symbols in the signature,
$\cM_1 = \{ F,\, \Pi \}$
and~$\cM_2 = \{ F,\, \Pi' \}$.
Then,
$\cI$ is a model of~$\Phi(\Pi_1) = F \wedge \Phi(\Pi)$
but not of $\Phi(\Pi_2) = F \wedge \Phi(\Pi')$.
Since $\Pi$ and~$\Pi'$ are strongly equivalent under context~$\Gamma$ and obviously~$F \models \Gamma$,
we get that $\Pi_1$ and~$\Pi_2$
have the same answer sets, so they have the same Herbrand models, which is a contradiction with the assumption.
\end{proof}

\end{document}